\documentclass[12pt,titlepage]{article}

\usepackage[titletoc,title]{appendix}
\usepackage[latin1]{inputenc}
\usepackage{amsmath}
\usepackage{amssymb}
\usepackage{graphicx}
\usepackage[margin=1.5in]{geometry}
\usepackage[onehalfspacing]{setspace}
\usepackage[bottom]{footmisc}
\usepackage{indentfirst}
\usepackage{endnotes}
\usepackage{rotating}
\usepackage{multirow}
\usepackage{booktabs,calc}
\usepackage{xcolor}

\usepackage{threeparttable}
\usepackage{float}
\usepackage{url}
\usepackage{amsmath,amsfonts}
\usepackage{amsthm}

\usepackage{natbib}
\usepackage{kbordermatrix}
\usepackage{verbatim}

\newtheorem{example}{Example}

\newcommand{\one}{\mathbf{1}}

\def\ngoods{m}
\def\nagents{N}

\makeatletter
\renewcommand\paragraph{\@startsection{paragraph}{4}{\z@}%
  {-3.25ex \@plus -1ex \@minus -0.2ex}%
  {0.01pt}%
  {\normalfont\bfseries\nobreak}%
}
\makeatother

\def\one{\mathbf{1}}

\def\argmax{\mbox{argmax}}

\def\la{\lambda}
\def\al{\alpha}

\def\w{\omega}

\def\ta{\theta}

\def\one{\mathbf{1}}

\def\E{\mathbf{E}}

\def\Re{\mathbf{R}}

\newcommand{\df}[1]{\textit{#1}}

\theoremstyle{plain}
\newtheorem{theorem}{Theorem}
\newtheorem{lemma}[theorem]{Lemma}
\newtheorem{proposition}[theorem]{Proposition}
\newtheorem{corollary}[theorem]{Corollary}

\begin{document}
\title{\large{Empirical Welfare Economics}\thanks{This paper is dedicated to the memory of Kim Border. We are grateful to  audiences at the CUHK-HKU-HKUST Joint Theory Seminar, UC Berkeley, the University of Chicago, The 2023 NSF-NBER Conference on Mathematical Economics at Caltech (which was celebrated in honor of Kim Border), McGill University, the 2022 Social Choice and Welfare meetings, Stanford University, the Workshop on Applications of Revealed Preferences, and Roy Allen for detailed comments.  Three anonymous referees and an editor also provided very useful feedback.}}

\author{ \small{Christopher P. Chambers}  \\
    \small{Department of Economics} \\
    \small{Georgetown University} \\
    \small{Christopher.Chambers@georgetown.edu}
    \and
    \small{Federico Echenique} \\
    \small{Department of Economics} \\
    \small{UC Berkeley} \\
    \small{fede@econ.berkeley.edu}
}
\date{\small{ \today }}

\maketitle

\begin{abstract}
Welfare economics relies on access to agents' utility functions: we revisit classical
questions in welfare economics, assuming access to data on agents' past choices
instead of their utilities.  Our main result considers the existence of utilities
that render a given allocation Pareto optimal. We show that a candidate allocation is
efficient for some utilities consistent with the choice data if and only if it is
efficient for an incomplete relation derived from the revealed preference relations
and convexity. Similar ideas are used to make counterfactual choices for a single
consumer, policy comparisons by the Kaldor criterion, and offer bounds on the degree
of inefficiency in a Pareto suboptimal allocation.
\end{abstract}

\section{Introduction}\label{sec:introduction}

Consider a social planner facing a collection of agents in a neoclassical resource allocation problem. Pareto optimality is characterized by the equality of agents' marginal rates of substitution, but to use this characterization our planner needs access to agents' utility functions. Suppose instead that the planner has access to a dataset consisting of a finite set of demand observations for each individual. The planner wants to know which allocations can be Pareto efficient for the collection of agents,  given what she knows from the observed dataset.  As a minimal discipline, she asks that there are monotone and convex preferences that are consistent with the data, and for which a given allocation is Pareto efficient. 

Our main result provides a complete characterization of the allocations that can be
Pareto efficient for the observed dataset. Our characterization parallels the
definition of Pareto optimality, with an empirical domination relation standing in
for unobservable utility comparisons. So the characterization says that there should
be no dominating alternative allocation, where the notion of domination captures what
can be inferred about agents' utilities from the dataset. In particular, the dataset
defines a revealed preference relation.  The revealed preference is, in general,
incomplete; it does not compare all alternatives.  Given revealed preference, we can
speak of making further comparisons based on monotonicity, transitivity, and
convexity.  For example, if it is known that both $x$ and $y$ are revealed preferred
to $z$, then $\frac{1}{2}(x+y)$ should also be at least as good as $z$.  Further,
imposing monotonicity allows for additional comparisons:  if $x$ is revealed
preferred to $z$, and $w \geq x$, then $w$ should also be preferred to $z$. Each such
comparison can be further combined with transitivity in order to impose additional
comparisons.  All the inferences that we can make recursively, using indirect revealed preference, convexity, and monotonicity, define what we call a domination relation for each individual agent.   This domination relation is, in a sense, the ``smallest'' set of inferences we can make from the data by using rationality, convexity and monotonicity alone.

The domination relation is typically highly incomplete. Incompleteness results from the limitations in the information contained in the data, even when augmented by the consequences of assuming monotone and convex preferences. This is in contrast with the normative statements about incomplete preferences, as in the work of \citet{ok2002utility,dubra2004expected,eliaz2006indifference}.  Efficiency with respect to our relation is the same notion as is used in the matching literature, where the incomplete relation is typically the stochastic dominance relation on a set of lotteries induced by a linear order on the set of degenerate outcomes.  See e.g. \citet{bogomolnaia2001new,mclennan2002ordinal, abdulkadiroglu2003ordinal,manea2008constructive, carroll2010efficiency, bogomolnaia2012probabilistic, hashimoto2014two, aziz2015universal, dogan2016efficiency}.


The paper actually uses the domination relation, and related concepts, to address a host of related questions in welfare economics. We start from individual welfare comparisons, and ask for counterfactual  comparisons that may be inferred from individual-level consumption data. In particular, given data from one consumer, and two new bundles $x$ and $y$, we ask when one can infer that the utility of $x$ is greater than that of $y$, for all rationalizing concave utilities. The exercise follows \cite{varian1982nonparametric}, and is related to the literature on demand bounds; see e.g. \cite{blundell2007improving,blundell2008best,blundell2012sharp,allen2020satisficing,allen2020counterfactual}. Our answer depends on a notion of empirical domination that is closely related to the notion behind our result on Pareto domination. There is, again, an empirically defined partial order among consumption bundles that captures all the comparisons that may be inferred from the dataset and the hypotheses of monotonicity, transitivity and concavity.

Next, we turn to collective welfare comparisons. Aside from our main result on Pareto optimal allocations, which we have already described, we consider the Kaldor criterion: whether an economic policy decision can be defended on the grounds that those who benefit from the policy could compensate those who lose \citep{kaldor1939welfare,hicks1939foundations,graaff1967theoretical}. Again the idea of domination gives us an answer, and serves to rule out whether demand data validates a policy decision. 

One approach \label{editor:oneapproach} to the problem could start from discrete choice. Imagine a single agent choosing from a finite set of alternatives. By observing choices from a set of possible finite menus, one could construct an incomplete revealed preference relation. Its transitive closure would in principle be incomplete, and a set of possible ``completions'' (or extensions) are possible. Now, with more than one consumer, and two alternative allocations $x$ and $y$, we can decide if $y$ might Pareto dominate $x$ by checking if there are completions for each consumer so that $y$ is ranked above $x$. This will occur as long as no agent's transitive closure ranks their consumption in $x$ above the one in $y$. Notice that this gives a nice answer in the discrete case when we only have a single competing allocation, $y$. In testing for Pareto optimality of $x$, however, we need to account for \textit{all} possible competing allocations. 

Our approach deals with the (neoclassical) model of infinitely divisible consumption:
Not a finite set of alternatives, and not discrete choice. The problem is handled by
an appeal to the ideas behind the second welfare theorem. An allocation is Pareto
optimal if and only if there is a common supporting price vector for each agent's
consumption. So we study a linear formulation of the problem of whether there exists
utilities that are consistent with the observed data, and that render a candidate
allocation Pareto optimal. Our theorem results from an application of linear programming duality.  As a consequence, the question is computationally tractable, and our conditions can be checked in ways that are computationally efficient. 

We focus on testing whether a given allocation could be Pareto optimal for some profile of utilities that are consistent with the data. We think of this as a natural and practical question that would come up in discussions of public policy. Consider a policy proposal that would result in an allocation $\bar x$. Can we say that the allocation $\bar x$, and by implication the underlying policy proposal, are possibly efficient? If the conditions we have laid out are violated, then there are not utilities for which the policy results in an efficient outcome for the economy.  
    
A more general question takes as given multiple allocations, and wonders if there is a utility profile that is consistent with the data and for which all the allocations under consideration are efficient. This more general question is interesting, but somewhat harder to motivate because it is not obviously tied to a given policy proposal.

When considering multiple allocations, our approach falls short of a full
characterization of efficiency, but still provides a practical and linear test.  That
is, if any of the allocations in a set of multiple allocations violates our
single-allocation condition, we know that the set as a whole cannot be possibly
efficient.  While we cannot quantify how often it is that each member of a set passes
the test while the set fails as a whole, clearly this is a nontrivial possibility (we
demonstrate an example in Section~\ref{sec:QL}). In the special case of quasi-linear
utility, we do offer a full characterization. 


\paragraph{Related Literature.}

The paper starts with a discussion of the individual welfare comparisons that may be inferred from a consumption dataset using revealed preference tools. Then the paper turns to collective choice. Our results on individual welfare extend the ideas of \cite{varian1982nonparametric}, who considered how two consumption bundles that are not observed in the data might be ranked by a utility that rationalizes the data. Varian provides an answer in terms of a system of linear inequalities. We show that the answers using his linear system is equivalent to checking a condition that is derived from the data.

Our results on collective choice fit into two strands of literature. First, 
the theory of efficiency in classical economic environments without completeness is studied in many works; a few of these include \citet{shafer1975equilibrium, gale1975equilibrium,gale1977role, fon1979classical, weymark1985remarks, rigotti2005uncertainty} \cite{bewley2002knightian}, and \cite{bewley1987knightian}. In our case, preference incompleteness arises because of limited data on agents' preferences, and gives rise to challenges that are not present in the previous literature.

Preference incompleteness goes away, and \label{ref:largedata} the results in our paper cease to be interesting, when agents' preferences can be recovered with high levels of precision from the observed data. The recovery question is, however, not straightforward; even when large consumption datasets are available. \cite{mas1977} discusses counterexamples, and conditions under which preferences may be recovered from a demand function, while \cite{mas1978revealed} shows that the canonical Afriat rationalization may (under a Lipschitz condition on the underlying demand behavior) be used to recover agents' preferences. See also \cite{chambersrecovering} and \cite{ugarte2022preference}. Among other conditions, these results require that the data sample a rich enough subset of the possible budgets. 

The second strand of literature concerns testing whether certain allocations can be equilibria of a given economy.  \citet{brown1996testable} are the first to formulate the problem as a revealed preference exercise.  In that paper, the authors check whether a collection of candidate objects could be equilibria of a given economy.  Results in the revealed preference literature usually focus on establishing a list of polynomial inequalities that must be satisfied in order for the data to be rationalizable---these inequalities are analogous to the ``Afriat inequalities'' of rational consumer behavior.  In showing that a particular rationalization problem reduces to one of verifying whether a solution exists to a list of polynomial inequalities establishes that these problems are decidable, in an algorithmic sense.  See also \citet{brown2000uniqueness,bossert2002core,kubler2003observable,carvajal2004equilibrium,carvajal2004testable,bachmann2004rationalizing,bachmann2006testable1,bachmann2006testable2,brown2007nonparametric,brown2008refutable,carvajal2010testable,cherchye2011testable,carvajal2018testing} for testable implications of related environments.  Some of these investigate efficiency directly:  \citet{bossert2002core} discuss  how the core correspondence varies (for fixed preferences) as endowments vary. Their results characterize the testable implications of the core, but is restricted to the case of two agents and a fixed aggregate endowment; their ``data'' is generated by varying the distribution of a fixed aggregate endowment.\label{ref:core}  \citet{bachmann2006testable1} considers an environment in which collections of endowments and consumption bundles (but not prices) are observed.  His Proposition 5 establishes that Pareto efficiency has essentially no testable content in this environment, even if all preferences are represented by strictly concave and continuously differentiable utilities.\footnote{The idea is that a common linear preference renders every allocation efficient.  Then perturb each agent's utility a bit to ensure strict concavity and smoothness.} 

\citet{allen2019revealed,allen2020satisficing,allen2020counterfactual} also consider notions of welfare or of group decision making.

As mentioned, when it comes to welfare comparisons, what these papers primarily do is provide an analogue of the result of \citet{afriat1967construction}, whereby rationalizability is equivalent to the satisfaction of a set of inequalities.  In contrast, our work differs in two respects:  first, we provide an economic characterization of whether a given bundle could possibly be efficient---our characterization is more analogous to the characterization of rationality via absence of cycles (also discussed by \citet{afriat1967construction}, and termed ``Generalized Axiom of Revealed Preference'' by \cite{varian1982nonparametric}).  We take as the starting point of our proof a collection of ``Afriat inequalities'' that must be satisfied, and use these to uncover a dual system of linear inequalities that we can interpret --- they have concrete economic meaning --- and deliver a condition in terms of the domination relation.  



\section{The model}\label{sec:themodel}

\paragraph{Basic definitions and notational conventions.}
We use the following notational conventions:  For vectors $x,y\in \Re^n$,  $x\leq y$
means that $x_i\leq y_i$ for all $i=1,\dots, n$;  $x < y$ means that $x\leq y$ and
$x\neq y$; and  $x\ll y$ means that $x_i < y_i$ for all $i=1,\dots, n$. The set of
non-negative vectors in $\Re^n$ is denoted $\Re^n_{+}$, and the set of vectors that
are strictly positive in all components is $\Re^n_{++}$. When $n$ is a non-negative
integer, we write the set $\{1,\ldots,n\}$ as $[n]$; with $[0]$ denoting the empty set.

A function $f:A\subseteq\Re^n\to\Re$ is \df{weakly monotone increasing}, or \df{non-decreasing}, if $f(x)\leq f(y)$ when $x\leq y$; and \df{monotone increasing}, if it is weakly monotone increasing and $f(x) < f(y)$ when $x\ll y$. We often just write ``increasing.''


A function $u$ is \df{explicitly quasiconcave} if it is quasiconcave and, for all $x,y\in \Re^n_+$ and $\la\in(0,1)$,  $u(x)\neq u(y)$ implies that \[ 
u(\la x + (1-\la) y)> \min\{u(x),u(y) \}.
\]  Observe that explicit quasiconcavity of $u$ is a behavioral property, meaning a property of the preference relation represented by $u$; and that it is weaker than concavity. Indeed, explicit quasiconcavity is only a minor strengthening of quasiconcavity; it is weaker than strict quasiconcavity ($u(\la x + (1-\la) y)> \min\{u(x),u(y) \}$ for all $\la\in (0,1)$), which corresponds to strict convexity of preferences. Strict quasiconcavity rules out that indifference curves contain any flat regions (i.e contain any line segments), but flat regions are allowed by explicit quasiconcavity (some rather pathological examples with flat regions are ruled out). Perhaps explicit quasiconcavity is best known because it ensures that local maxima are global maxima, for which quasiconcavity alone does not suffice (see Theorem~192 in \cite{BorderOPtim}).

\paragraph{Definitions from welfare economics.} 
An agent is defined through a preference relation on $\Re^m_+$, which we represent throughout by a utility function $u:\Re^m_+\to\Re$.\footnote{We restrict attention to continuous preference relations, but given that preferences are only constrained to rationalize a finite dataset, continuity is without loss of generality.} The elements of $\Re^m_+$ are called \df{consumption bundles}. Given a finite set  of agents $N$, an \df{allocation} is a vector $\bar x=(\bar x_i)_{i\in N}\in\Re^{mN}_+$.\footnote{One should think of an allocation $\bar x$ as ``allocating'' the aggregate bundle $\sum_{i\in N} \bar x_i$ among the agents in $N$.} If each agent is endowed with a utility function $u_i$, an allocation $\bar y$ \df{Pareto dominates} the allocation $\bar x$ if $u_i(\bar y_i)\geq u_i(\bar x_i)$ for all $i$, with a strict inequality for at least one agent. An allocation $\bar x$ is \df{Pareto optimal} if there is no allocation satisfying \[\sum_{i\in N} \bar y_i = \sum_{i\in N}\bar x_i\] that Pareto dominates it.

Next we turn to a criterion for comparing allocations based on the principle that winners may compensate the losers. The idea is that those who gain in moving from one allocation to the other may compensate those who lose with the move in allocations.  Let $\bar x$ and $\bar y$ be two allocations. Say that $\bar x$ \df{weakly Kaldor dominates} $\bar y$ if there is no allocation $\bar z$ with $\sum_i \bar z_i\leq \sum_i \bar y_i$ that Pareto dominates $\bar x$. The idea is that if $\bar x$ does not weakly dominate $\bar y$, then there is a way of re-assigning (whence losers are compensated by winners) the aggregate bundle $\sum_i \bar y_i$ in a way that Pareto dominates $\bar x$ (see Chapter 5 in \cite{graaff1967theoretical} for a discussion of the Kaldor criterion).



\paragraph{Data and rationalizability.} 
A pair $(p,x)\in\Re_+^{m + m}$ is an \df{observation}, and should be interpreted as
the datum that the consumption bundle $x\in\Re_+^m$ was chosen from the budget set
$\{y\in\Re^m_+ : p\cdot y\leq I \}$ in which the income, or budget, is $I=p\cdot
x$. A (possibly empty) finite list of observations $(p^k,x^k)_{k\in [K]}$ is termed an \emph{individual dataset}.  $N$ is a finite set of individuals. A \emph{group dataset} is a collection of individual datasets, one for each $i\in N$.  So, $D_i = \{(p_i^k,x_i^k)\}_{k\in [K_i]}$ denotes an individual dataset for individual $i$, and $\{D_i :i\in N\}$ is a group data set.

An individual dataset is \emph{rationalizable} if there is an increasing utility function $u_i:\Re^m_+\rightarrow \Re$ for which for all $k$, $u_i(x)>u_i(x_i^k)$ implies $p_i^k\cdot x > p_i^k \cdot x_i^k$.  In this case, we say that $u_i$ \df{rationalizes} the individual dataset (or that it is a \df{rationalizing} utility, when the dataset is implied).  Similarly, we say that a group dataset is \df{rationalizable} if each individual dataset is rationalizable.

In our paper, we insist that rationalizing utilities be monotone increasing. Clearly, some structure must be assumed on utilities, or any data becomes rationalizable by a constant utility. The most common approach is to impose local non-satiation, and then resort to Afriat's theorem which says that one may without loss of generality assume a rationalizing utility that is both increasing and concave. Thus monotonicity, but more importantly concavity, comes for free in the case of an individual agent's observed consumption behavior. 

Revealed preferences involve the use of two binary relations. The \emph{direct revealed preference} of agent $i$ is denoted by $\succeq^R_i$, and defined by $x \succeq^R_i y$ if $x\geq x_i^k$ for some $k$ that satisfies $p_i^k \cdot x_i^k \geq p_i^k \cdot y$, or if $x=y$. The \emph{direct strict revealed preference} of agent $i$ is denoted by $\succ^R_i$, and defined by  $x \succ^R_i y$ if  
\[x\gg x' \succeq^R_i y, \text{ or } x\geq x^k_i \text{ and } p^k_i \cdot x^k_i > p^k_i\cdot y,
\] for some $x'$ or observation $k$.  These definitions of revealed preferences are slightly unusual, in that they already incorporate the expectation of a monotone preference, and symmetry is built-in.\footnote{See \cite{chambers2009supermodularity} and 
\cite{nishimura2017comprehensive} for such ``compositions'' of the revealed preference relation with the partial order on consumption bundles. It is easy to see that Afriat's theorem remains true under our definition of revealed preference.} Observe that ${\succ^R_i} \subseteq {\succeq^R_i}$.

The \emph{indirect revealed preference} $\succeq^I_i$ is defined as the transitive closure of $\succeq^R_i$.  The \emph{indirect revealed strict preference} $x \succ^I_i y$ obtains when there is a finite chain $x=z_1 \succeq^R_i \ldots \succeq^R_i z_L=y$, where at least one instance of $\succeq^R_i$ is $\succ^R_i$.

An individual dataset $D_i$ satisfies the \df{Generalized Axiom of Revealed Preference (GARP)} if there is no $x,y\in\Re^m_+$ such that $x\succeq^I_i y$ while $y\succ^I_i x$.

\section{Results}\label{sec:results}

We consider counterfactual welfare comparisons. Given data on individual consumption, we seek to characterize which counterfactual (i.e.\ unobserved) welfare conclusions may be drawn on the basis of what can be inferred about agents' preferences from the data.  For individual agents, we want to evaluate unobserved bundles. For a group of agents, the welfare comparisons are about the possible Pareto optimality of some allocation, or consistency with the Kaldor criterion. 

All proofs are relegated to Section~\ref{sec:proofs}.

\subsection{Individual welfare} \label{sec:individualwelfare}

We begin by discussing individual welfare conclusions that may be drawn from a single agent's consumption dataset. Aside from the intrinsic merit of these results, they serve to introduce some of the ideas we use later in our (main) results on collective welfare.

Our first result asks when we can say that one bundle is unambiguously better than another, given what the data tell us about the agent. Specifically, given an individual dataset $\{(x^k,p^k):1\leq k\leq K \}$, and two bundles $\bar x$ and $\bar y$, when is $\bar x$ ranked above $\bar y$ for all increasing and concave utility functions compatible with the data? 

The answer turns out to depend on a binary relation that may be inferred from the consumer's choices. \cite{varian1982nonparametric} also considers this question and offers an answer in the form of a linear program; what Varian calls Fact 4. Our binary relation essentially emerges from the dual program to Fact 4. Say that $\bar x$ \df{bests} $\bar y$ if $\bar x$ is a convex combination of some collection $z^l$ of bundles, $1\leq l \leq L$, such that, for each $l$,  $z^l\succeq^I \bar y$. A bundle $\bar x$ \df{strictly bests} $\bar y$ for agent $i$ if it weakly bests it and, moreover, if in the defining convex combination there is $l$ with $z^l\succ^I \bar y$.\footnote{A bundle $\bar x$ strictly bests itself when it is incompatible as a choice with the existing dataset. This means that there is no price $\bar p$ at which $\bar x$ could be demanded, and for which the resulting dataset (obtained by adding $(\bar x,\bar p)$  to the  dataset) is rationalizable. If the dataset is rationalizable, however, we may choose $\bar p$ that supports the upper contour set of a (without loss, concave) rationalizing utility at $\bar x$. Adding the resulting observation to the dataset preserves its rationalizability.}

It is easy to see that if $\bar x$ strictly bests $\bar y$, then it is ranked above $\bar y$ by any rationalizing concave and monotone increasing utility function $u$. Indeed, if $\bar x =\sum_l \la_l z^l$ is as above, then:
\begin{align*}
    u(\bar x) & \geq \sum_{l=1}^L \la_l u(z^l) \\
    & >  \sum_{l=1}^L \la_l u(\bar y) = u(\bar y).
\end{align*} 
The first inequality follows from concavity, and the second from $u$ rationalizing the data and the requirements on $z^l$ in the definition of besting. Our first result says that strict besting is not only sufficient for the counterfactual comparison of two bundles, but also necessary.

\begin{theorem}\label{thm:varianwelfare}
Let  $(x^k,p^k)_{1\leq k\leq K}$ be an individual  dataset and $\bar x, \bar y\in\Re^m_+$ be two bundles.  Then $u(\bar x)> u(\bar y)$ for all concave and monotone increasing $u$ that rationalize the dataset if and only if $\bar x$ strictly bests $\bar y$.
\end{theorem}

In Theorem~\ref{thm:varianwelfare}, we require that every concave, increasing, rationalizing utility satisfies a certain property. In our next result, Theorem~\ref{thm:individualwelfare}, we asks about the existence of a rationalizing utility with a certain property. The latter sort of result is, of course, most conclusive when the condition fails, and thus certifies that the property is incompatible with any rationalizing utility. Our main results in Section~\ref{sec:collectivewelfare} are of this nature. 

Finally, observe that Theorem~\ref{thm:individualwelfare} only asks utilities to be explicitly quasiconcave. The same will be true of our main results.

\begin{theorem}\label{thm:individualwelfare}
Let $(x^k,p^k)_{1\leq k\leq K}$ be an individual dataset and $\bar x\in\Re^m_+$ a bundle. There exists a monotone increasing and explicitly quasiconcave rationalizing utility $u$ for which $u(\bar x)\geq \max \{u(x^k):1\leq k\leq K \}$ if and only if, once we add $\bar x\succeq^R x^k$ for all $k$ to the revealed preference relation, as well as 
as well as $x^k \succeq^R \bar x$ when $p^k \cdot (\bar x - x^k) \leq 0$ and $x^k \succ^R \bar x$ when $p^k \cdot (\bar x - x^k) <0$, we have
\begin{enumerate}
    \item GARP is satisfied.
    \item There is no bundle $y\leq \bar x$ that strictly bests $\bar x$. 
\end{enumerate}
\end{theorem}

\subsection{Collective welfare}\label{sec:collectivewelfare}

Our main result characterizes the allocations that are efficient for some utility functions (with the requisite properties) that are consistent with a group dataset. 

An allocation $\bar y$ \df{empirically dominates} the allocation $\bar x$ if $\sum_i \bar y_i \leq \sum_i \bar x_i$ while $\bar y_i$ bests $\bar x_i$ for all $i$, and strictly bests it for at least one~$i$. Observe the parallelism with the notion of Pareto domination: Given increasing utility functions $(u_i)_{i\in N}$ we may say that an allocation $\bar y$ Pareto dominates $\bar x$ if $\sum_i \bar y_i \leq \sum_i \bar x_i$, while $u_i(\bar y_i)\geq u_i(\bar x_i)$ for all $i$, and $u_i(\bar y_i)> u_i(\bar x_i)$ for at least one $i$. Theorem~\ref{thm:welfare} implies that empirical domination really is the empirical counterpart to Pareto domination. 

\begin{theorem}\label{thm:welfare} 
Let $(D_i)_{i\in N}$ be a rationalizable group dataset, and $\bar x$ an allocation. The following statements are equivalent:
\begin{enumerate}
    \item\label{it:wf1} There are monotone increasing, and explicitly quasiconcave, rationalizing utilities for which $\bar x$ is Pareto efficient.
    \item\label{it:wf2} There are monotone increasing, and concave, rationalizing utilities for which $\bar x$ is Pareto efficient.
    \item\label{it:wf3} The allocation $\bar x$ is not empirically dominated by any other allocation. 
\end{enumerate}
\end{theorem}

The theorem provides a characterization of the allocations that could be efficient, for some monotone and convex preferences of the agents (with the minor strengthening of convexity implied by explicit quasiconcavity). The role of the unobserved utility functions in the definition of Pareto domination is taken by the observable empirical domination relations.\footnote{The proof of Theorem~\ref{thm:welfare} is based on an application the theorem of the alternative.  A different method of proof would be to construct the revealed preference relations of the candidate bundle $\overline{x}$ which is not empirically dominated, and then attempt to separate the implied Scitovsky set of this bundle from the set of bundles $y$ for which $\sum_i \overline{x}_i \gg y$, resulting in a supporting price $p$.  The idea would then be to show that adding, for each agent, the observation $(p,\overline{x}_i)$ results in a new dataset for each agent, where each new dataset satisfies GARP.  Since the price $p$ is common to all agents, there is a common marginal rate of substitution for any preference rationalization, and so we would have an efficient bundle.  Though this method is certainly more intuitive than what we have done, we were unable to show in general that these new datasets generally satisfy GARP without reverting to the theorem of the alternative.} 


Empirical domination ensures the existence of a common supporting price at the allocation $\bar x$, essentially the equality of marginal rates of substitution for a collection of rationalizing utilities. If we additionally require that this price supports the \df{Scitovsky contour} at $\bar x$, then the ideas behind Theorem~\ref{thm:welfare} can be used to provide an empirical basis for the Kaldor criterion:\footnote{Given utilities $(u_i)$, the \df{Scitovsky contour} at $\bar x$ is the set $S(\bar x) = \{ 
\sum_i z_i : u_i(z_i)\geq u_i(\bar x_i) \text{ for all } i\in N\}$. If a price $q$ supports all individual upper contour sets at $\bar x$ and $q\cdot \sum_i \bar y_i < q\cdot \sum_i \bar x$, then $\sum_i\bar y_i\notin S(\bar x)$.}

\begin{corollary}\label{cor:kaldor}
Let $(x_i^k,p_i^k)_{1\leq k\leq K_i}$, for $i\in N$, be a rationalizable group dataset. Let $\bar x$ and $\bar y$ be allocations. There are increasing, concave, rationalizing utilities for which $\bar x$ weakly Kaldor dominates $\bar y$ if there is no allocation $(\bar z_i)$ that weakly dominates $\bar x_i$ for all $i$, and strictly dominates it for at least one $i$, and a scalar $\kappa\geq 0$, for which 
\[ 
\sum_i \bar z \leq \sum_i \bar x_i + \kappa (\sum_i \bar y_i - \sum_i \bar x_i)
\]
\end{corollary}

Observe that Corollary~\ref{cor:kaldor} only offers a sufficient condition for Kaldor domination. When the condition holds, then we may say that there are rationalizing utilities for which a switch from $\bar x$ to $\bar y$ could not be defended on the basis of the Kaldor criterion. 

Our results assume that consumers' datasets are rationalizable. Empirical studies often document violations of this property, but there is (arguably) evidence of many environments where such violations are relatively small. \label{page:RefCGARP} See, for example, \cite{echenique2011money} and the summary of the empirical literature discussed in \cite{chambers2016revealed}. Welfare comparisons for inconsistent agents is, in any case, a conceptually challenging question.

\section{Multiple allocations}\label{sec:QL}

The results obtained in Section~\ref{sec:results} exemplify the power of our approach, but there are also clear limits. Given a dataset, one may ask a related question for a \emph{collection} of allocations: whether there exists a single economy capable of generating all such allocations as Pareto efficient ones.  It is natural to conjecture that there is such an economy if and only if each of the allocations is undominated.  This conjecture turns out to be false, as shown by the following example: 

\begin{example}\label{ex:multiplealloc}Let $N=\{1,2\}$, and suppose there are two commodities, so that $m=2$.  Individual $1$ has an empty individual dataset.  Individual $2$ has four observations:  $(p_2^1,x_2^1) = ((2,1),(1,2))$, $(p_2^2,x_2^2)=((2,1),(0,4))$, $(p_2^3,x_2^3)=((1,2),(2,1))$, and $(p_2^4,x_2^4)=((1,2),(4,0))$.

Now, suppose we want to consider the allocations $\bar x_1^1 = (1,0)$, $\bar x_2^1 = (0,4)$, and $\bar x_1^2 = (0,1)$, $\bar x_2^2 = (4,0)$.  Observe that because individual $1$ has an empty individual dataset, each of these allocations are possibly efficient by Theorem~\ref{thm:welfare}.  On the other hand, they cannot both be efficient for the same economy.  To understand why, observe that if $q^1$ supports $x_2^1$, then $q^1 \cdot (0,4) \leq q^1 \cdot (1,2)$,  as the individual data set for individual $2$ is rational.  If $q^1(2) = 0$ (the second coordinate of $q^1$), then this inequality is obviously strict as $q^1 \geq 0$. 

So, if $q^1(2) = 0$, we conclude that $q^1 \cdot (1,2) - q^1 \cdot (0,4) > 0$, so that $q^1\cdot (1,-2)>0$, from which we conclude $q^1\cdot (1,-1)>0$, or $q^1 \cdot x_1^1 > q^1 \cdot x_1^2$.  Similarly, if $q^1(2) >0$, then we know $q^1 \cdot (1,-2) \geq 0$, so that (as $q^1(2) > 0$), $q^1 \cdot x_1^1 > q^1 \cdot x_1^2$.

So, $q^1 \cdot x_1^1 > q^1 \cdot x_1^2$; symmetrically, $q^2 \cdot x_1^2 > q^2 \cdot x_1^1$.  These inequalities obviously cannot simultaneously hold for a rational decision maker.
\end{example}

In the proof of Theorem~\ref{thm:welfare}, we reduced the problem of testing whether an allocation $\bar x$ could be efficient to the question of the existence of a supporting price $q$. Were we to ask that multiple allocations be efficient, we would need different supporting prices for each such allocation; but more to the point, the scale factors could differ across individuals, thus rendering the system nonlinear. In other words, we would need different $\lambda$ for the different allocations, and the normalization used in the proof of Theorem~\ref{thm:welfare} would no longer work. 

The problem goes away when the value of $\la$ is fixed, suggesting that we 
consider the special case of quasilinear utility: when $\la=1$. In this case, we can indeed
characterize the collections of allocations that may jointly be Pareto efficient for
some rationalizing quasilinear utility. We first revisit the rationalizability
question for a single consumer,  a question first analyzed by
\citet{brown2007nonparametric}, and then turn to the problem of a collection of
allocations.

\subsection{Individual data}
Consider an individual data set
$(x^k,p^k)_{1\leq k\leq K}$. We say that it is \df{quasi linear rationalizable} if there exists a utility function $U:\Re^\ngoods_+\to\Re$ so that, for all $k\in[K]$,
\[
U(x^k) - p^k\cdot x^k \geq U(x) - p^k\cdot x
\] for all $x\in\Re^n$.

A matrix $\eta\in\mathbb{R}_+^{K\times K}$ is \emph{bistochastic} if for every $t\in [K]$, $\sum_{s\in [K]} \eta(t,s)=\sum_{s\in [K]} \eta(s,t)=1$.

The following result is a form of the theorem in \citet{brown2007nonparametric},
essentially an economic analogue of the notion of \emph{cyclic monotonicity} due to
\citet{rockafellar1966characterization}.\footnote{See also \citet{browning1989nonparametric}.}  We state it without proof (its proof is implicit in the proof of the following theorem as well). 

\begin{proposition}\label{prop:QL}
An individual dataset $(x^k,p^k)_{1\leq k\leq K}$ is quasi linear rationalizable by a concave and monotonic utility if and only if, for any bistochastic matrix $\eta\in\mathbb{R}_+^{K\times K}$, we have $\sum_{k}\sum_t \eta(k,t) p^t\cdot (x^k-x^t) \geq 0$.
  \end{proposition}

To interpret Proposition~\ref{prop:QL}, think of a bistochastic matrix as a probability distribution over pairs $(k,t)$, after a normalization. If $U$ is a rationalization of the data, then the sum $\frac{1}{\sum_{(k,t)}\eta(k,t)}\sum_{(k,t)}\eta(k,t) [U(x^k)-U(x^t)]$ is the expected change in utility when going from the consumption $x^t$ to $x^k$. If the matrix is bistochastic, this expected change is zero. On the other hand, since $U$ rationalizes the data, for each $k$ and $t$,
$U(x^t) - p^t\cdot x^t \geq U(x^k) - p^t\cdot x^k$. Thus then change in utility
$U(x^k)-U(x^t)$ is bounded above by $p^t\cdot (x^k-x^t)$. And therefore the expected value of $p^s\cdot (x^k-x^t)$, $\frac{1}{\sum_{(k,t)}\eta(k,t)}\sum_{(k,t)}\eta(k,t) p^t\cdot (x^k-x^t)$ must be non-negative.

\subsection{Multiple allocations}

Here, we will show that the quasi-linear model allows a natural linear test of the
hypothesis that multiple allocations could potentially be Pareto
efficient. Example~\ref{ex:multiplealloc} shows that, in general, an allocation-by-allocation approach does not capture all of the implications imposed by hypothesizing that multiple allocations are Pareto efficient.  In the general setting, there is no linear test that we could perform.  But in the quasi-linear setting, it becomes quite simple.

Let $x^s=(x^s_1,\ldots,x^s_{\nagents})\in\Re_+^{L\nagents}$ for $s\in[L]$ be a collection of $L$ allocations. In a notational abuse, we regard the elements of $[K_i]$ and $[L]$ as distinct, even if they are the same number.

In the following, a matrix $\eta$ is constant row-column sum if there is some number
$c$ such that, for every $k$ and $l$, $\sum_t \eta(t,k)=c=\sum_t \eta(l,t)$.  That is, if it is a scaled version of a bistochastic matrix.
 
\begin{theorem}\label{thm:QL}There exist concave utilities $U_i$ that quasi-linear rationalize the data, and for which the allocations $x^s$ are Pareto optimal if and only if there are no constant row-column sum matrices $\eta_i\in\mathbb{R}_+^{([K_i]\cup [L])\times ([K_i]\cup [L])}$ for which:
\[\sum_i \sum_{t\in K_i}\sum_{k\in [K_i]\cup [L]}\eta_i(k,t)p_i^t\cdot(x_i^t-x_i^k)>0\] and for all $s\in L$, \[\sum_i \sum_{k\in [K_i]\cup [L]}\eta_i(k,s)(x_i^s-x_i^k)\geq 0.\]\end{theorem}

The idea follows as in Proposition~\ref{prop:QL}.  Suppose the data are rationalizable, and suppose $\eta_i$ satisfies the conditions in the theorem, and that each $x^s$ could be efficient.  We will argue that a contradiction entails.

Hypothesizing that for each $s$, $x^s$ is an efficient allocation means that there
are prices $q^s$ at which agent $i$ demands $x_i^s$.  Let us for now set $p_i^s=q^s$
for each $i\in N$ and $s\in L$.  Then, owing to Theorem~\ref{thm:QL}, we must
have \begin{equation}\label{eq:one}\sum_i \sum_{t\in [K_i] \cup [L]}\sum_{k\in [K_i]
    \cup [L]}\eta_i(k,t)p_i^t \cdot (x_i^t-x_i^k)\leq 0.\end{equation}  Now, since
for any $s\in L$, $\sum_i \sum_{k\in [K_i]\cup [L]}\eta(k,s)(x_i^k-x_i^s) \leq 0$ and
since $p_i^s=q^s\geq 0$, using linearity we have \begin{equation}\label{eq:two}\sum_i
  \sum_{k\in [K_i]\cup [L]}\eta(k,s)p_i^s\cdot (x_i^s-x_i^k)\geq 0\end{equation} for
  any $s\in [L]$.  For each $s$, subtracting equation~\eqref{eq:two} from
  \eqref{eq:one}, we get that $\sum_i \sum_{t\in K_i}\sum_{k\in K_i\cup
    Ly}\eta(k,t)p_i^t\cdot (x_i^t-x_i^k) \leq 0$, contradicting the first equation in
  the statement of the theorem. We offer a formal proof in Section~\ref{sec:proofthmQL}.


To interpret the conditions in Theorem~\ref{thm:QL}, we may assume that each $\eta_i$ is actually a bistochastic matrix, and, by renormalizing, that there is some probability distribution $\alpha\in \Delta(N)$ such that the equations in the Theorem may be rewritten as:
$$\sum_i  \sum_{t\in [K_i]}\sum_{k\in [K_i] \cup [L]}\al_i\eta_i(k,t)p_i^t \cdot
(x_i^k-x_i^t) < 0$$
and for each $s\in [L]$, 
\[\sum_i \sum_{k\in [K_i]\cup [L]} \al_i \eta_i(k,s)(x_i^k-x_i^s)\leq 0.\]
We shall see that the conditions in the theorem are, roughly speaking, multi-agent
analogues of the conditions for quasi-linear rationalizability in
Proposition~\ref{prop:QL}, once we hypothesize common supporting prices for the
allocations $x^s$, $s\in[L]$.

Interpret the product $\al_i\eta_i(k,t)$ as a probability: draw an agent at random according to
$\al$, and then a pair $(k,s)$ using the bistochastic matrix. Just like in
Proposition~\ref{prop:QL}, the change in utility from $x^t_i$ to 
$x^k_i$ is upper bounded by $p^t\cdot(x^k_i-x^t_i)$. In a bistochastic matrix, the
expected utility change must be zero, and therefore  Proposition~\ref{prop:QL}
results as the expectation of $p^t\cdot(x^k_i-x^t_i)$ cannot be negative.

Theorem~\ref{thm:QL} is, however, about  efficiency, which demands that we
find supporting prices $q^s$ for each allocation 
$s\in [L]$. We may set individual prices $p^s_i=q^s$, because efficiency requires
that the same prices support each individual agent's consumption (a generalization of
the equalization of marginal rates of substitution).  Now, since $q^s>0$, if 
 $\sum_i\sum_{k\in [K_i]\cup [L]}\al_i \eta_i(k,s)(x_i^k-x_i^s)< 0$ holds for each
$s\in [L]$, then we obtain 
\[\sum_{s\in [L]} q^s\cdot \sum_i \sum_{k\in [K_i]\cup [L]}\al_i\eta_i(k,s)(x_i^k-x_i^s)
= \sum_i \sum_{s\in [L]}\sum_{k\in [K_i]\cup [L]}\al_i\eta_i(k,s)p^s_i\cdot (x_i^k-x_i^s)
< 0.\]

Using the upper bound on utility changes that we used in Proposition~\ref{prop:QL}, 
this means that the expected change in utility is negative (when drawing an agent at
random, and a pair of allocations from $[L]$ and $[K_i]\cup [L]$). But the overall
expected change must be zero, so the inequality in the formula is inconsistent with
efficiency. 

\section{Inefficient allocations}\label{sec:inefficient}

We now turn to an empirical evaluation of potentially inefficiency allocations. In particular, we present some results using the measure of Pareto inefficiency proposed by \cite{debreu1951coefficient}: the \df{coefficient of resource utilization}.

In order to introduce the relevant concepts, consider an allocation $x=(x_i)_{i\in N}$ and fix a profile of utility functions $(u_i)_{i\in N}$ for the agents in $N$. Let $S^{u}(x_i)=\{z_i\in\Re^\ngoods_+:u(z_i)\geq u(x_i) \}$ denote the \df{upper contour set} for
utility $u$ at consumption vector $x_i$; and 
\[
S^{u_1,\ldots,u_{\nagents}}(x_1,\ldots,x_\nagents) = \sum_{i\in N}S^{u_i}(x_i)
\] the \df{Scitovsky contour} at $x$ for the profile of utility functions
$(u_1,\ldots,u_{\nagents})$. In words, the Scitovsky contour of an allocation $x$ is the
set of aggregate bundles that may be decomposed into an allocation that guarantees
each agent at least the utility that they obtain in $x$.

Debreu observes that, if the allocation $x$ is not Pareto optimal, then as an aggregate consumption bundle, $\sum_{i=1}^\nagents x_i$ will lie in the interior of the Scitovsky contour $S^{u_1,\ldots,u_{\nagents}}(x_1,\ldots,x_\nagents)$. Debreu proposes to measure the degree of inefficiency in $x$ by the distance between $\sum_{i=1}^\nagents x_i$ and the boundary of the Scitovsky contour: essentially his measure quantifies the degree to which agents' implied welfare in $x$ can be reached with fewer resources than the aggregate $\sum_i x_i$.

Debreu's definition involves a price-dependent notion of welfare, but he shows that it reduces to a \df{coefficient of resource utilization} $\rho$ defined by 
\[
\rho = \inf\{\rho'\in [0,1]: \rho' \sum_{i=1}^\nagents x_i\in  S^{u_1,\ldots,u_{\nagents}}(x_1,\ldots,x_\nagents)\}.
\] We refer to \cite{debreu1951coefficient} for further details on his result (which requires convexity, continuity and monotonicity on agents' preferences). 

Now it should be clear that calculating the coefficient of resource utilization requires access to agents' utility functions. In our case, we use data on agents' consumption choices to obtain bounds on the possible values of the coefficient. In particular, consider a group data set: $D_i = \{(p_i^k,x_i^k)\}_{k\in [K_i]}$. Suppose, just to simplify our notation, that $K=K_i$ for all $i\in N$. We focus on the $K$th allocation $x^K=(x_1^K,\ldots,x_\nagents^K)$, and want to measure its degree of inefficiency by means Debreu's coefficient. 

A canonical utility rationalization in revealed preference theory is Afriat's construction. For each individual data set $D_i$, we may let the set $A_i\in\Re^{2K}$ consist of all vectors $(V_i,\la_i)=((V^1_i,\la^1_i),\ldots,(V^K_i,\la^K_i))$ that 
solve the Afriat inequalities for $i$'s data $D_i$, and that satisfy $V^K_i=1$ and
$\min\{V^k_i-\la^k_ip^k_ix^k_i:1\leq k\leq K\}=0$. Now we may define the \df{Afriat rationalization} $u^{(V_i,\la_i)}:\Re^\ngoods_+\to \Re$ by
\[
u^{(V_i,\la_i)}(x_i) = \inf\{V^k_i + \la^k_i p^k_i\cdot (x_i-x^k_i) :1\leq k\leq K \}
\] for each $(V_i,\la_i)\in A_i$. One bound on the coefficient of resource utilization us obtained by 
\[ 
\underline \rho = 
\inf\{\rho'\in [0,1]: \rho' \sum_{i=1}^\nagents x^K_i\in 
\sum_{i=1}^\nagents
\bigcup\limits_{(V,\la)\in A_i} S^{u^{(V,\la)}}(x^K_i)\}.
\]

Another bound is found by means of the utility $u^*_i$.
\[
u^*_i(x) = \inf\{u^{(V,\la)}(x):(V,\la)\in A_i \},
\]  note that $u^*_i$ is a rationalization of the data $D_i$. We may use these utilities to define a bound
\[ 
\bar  \rho = 
\inf\{\rho'\in [0,1]: \rho' \sum_{i=1}^\nagents x_i\in S^{u^*_1,\ldots,u^*_\nagents}(x^K_1,\ldots,x^K_\nagents)\}.
\]

\begin{proposition}\label{prop:CRU}
Consider a group dataset $D_i = (p_i^k,x_i^k)_{k\in[K]}$, $i\in [\nagents]$ in which each individual dataset $D_i$ is rationalizable. The coefficient of resource utilization from any profile of Afriat rationalizations is bounded above by $\bar \rho$. The coefficient of resource utilization from any profile of concave and monotone rationalizations is bounded below by $\underline\rho$.
\end{proposition}

\section{Remarks}\label{sec:remarks}

The key to our results is an observation based on Afriat's theorem, which says that an individual dataset $\{(p^k_i,x^k_i):1\leq i\leq K_i \}$ is rationalizable if and only if there is a solution $U^k_i,\la^k_i>0$ to the following system of linear ``Afriat inequalities:''\footnote{See \cite{chambers2016revealed} for a discussion of Afriat's theorem and this system of linear inequalities.}
\[ 
U^l_i \leq U^k_i + \la^k_i p^k_i \cdot (x^l_i - x^k_i).
\]
The observation is that we may normalize such a solution so that $\la^{k^*}_i=1$ for some specific observation $k^*$. As a result we obtain a system that remains linear, even if the prices $p^{k^*}_i$ at this particular observation were unknown. 

With this observation in hand, we can now approach a problem like that in Theorem~\ref{thm:welfare}. For the allocation $\bar x$ to be Pareto optimal, agents' utilities would need to have a common supporting price $q$ at $\bar x_i$. The existence of such a price $q$ may be added to the above system of inequalities as if it were a new observation. Assuming that the corresponding value of $\la$ has been normalized to 1, the system is still linear. See \citet{bachmann2004rationalizing} or \citet{bachmann2006testable1} for related constructions.  Now the work in proving the theorem amounts to interpreting the dual linear system.


We have discussed some obvious limits to our approach. Perhaps the main limitation is
that the rationalizing utilities may not be unique, leading to an indeterminacy when
the condition in our theorem is satisfied. But there are also additional applications
that we have not exhausted. One of these is envy-freeness. Suppose given a group
dataset, and consider the existence of rationalizing utilities that render some
proposed allocation $\bar x$ envy-free: meaning rationalizing utilities $(u_i)$ with
the property that $u_i(\bar x_i)\geq u_i(\bar x_j)$ for all $i,j\in N$. Our methods,
based on working through the dual of \textit{augmented} system of Afriat
inequalities, provide an answer to this question. 

A sketch of the solution follows: the trick is to add supporting prices for each agent at the proposed consumption of other agents in the allocation $\bar x$. The normalization idea keeps the system linear, and we just need to include utility values $u_{i,j}$ for $i$'s utility at the bundle intended for $j$:

\begin{enumerate}
\item For all $i\in N$ and all $k,l\in[K_i]$ for which $p_i^l \cdot (x_i^k -x_i^l)\leq 0$, we have $u_i^k\leq u_i^l + \lambda_i^l p_i^l \cdot (x_i^k-x_i^l)$.
\item For all $i,j\in N$ and all $k\in [K_i]$ for which $p_i^k \cdot (\overline{x}_j-x_i^k)\leq 0$, we have $u_{i,j} \leq u_i^k + \lambda_i^k p_i^k\cdot (\overline{x}_j - x_k^i)$.
\item For all $i,j\in N$ and all $k\in [K_i]$, $u_i^k \leq u_{i,j} + p_{i,j}\cdot (x_k^i-\overline{x}_j)$.
\item For all $i,j,h\in N$, $u_{i,j}\leq u_{i,h}+p_{i,h}\cdot (\overline{x}_j - \overline{x}_h)$.
\item For all $i,j\in N$, $u_{i,i}\geq u_{i,j}$.
\end{enumerate}

We omit the details, but hope that it is clear how to proceed on the basis of this system.



\section{Proofs}\label{sec:proofs}

\subsection{Proof of Theorem~\ref{thm:varianwelfare}}

In proving Theorem~\ref{thm:varianwelfare}, we shall make use of an auxiliary ``besting'' definition: Say that $\bar x$ \df{bests'} $\bar y$ if $\bar x$ can be written as a convex combination of bundles $z^l$, where for each $l$ 
$z^l \succeq^I \bar x,$ or $z^l \succeq^I \bar y$, with at least one occurrence of the latter. Say that $\bar x$ \df{strictly bests'} $\bar y$ if it weakly bests it, and one of the revealed-preference comparisons is strict ($\succ^I$ for $\succeq^I$).

\begin{lemma} \label{lem:besting}
If $\bar x$ bests' $\bar y$, then $\bar x$ bests $\bar y$. And if $\bar x$ strictly bests' $\bar y$, then $\bar x$ strictly bests $\bar y$. 
\end{lemma}
\begin{proof}
Suppose that $\bar{x}$ bests' $\bar{y}$ and by means of contradiction that $\bar{x}$ does not best $\bar{y}$.  

We can express $\bar{x} = \sum_k \mu^k w^k + \sum_l \lambda^l z^l$, where
$\mu^k,\la^l\geq 0$, $\sum_k\mu^k+\sum_l \la^l=1$,
each $w^k \succeq^I \bar{x}$ but not $w^k \succeq^I \bar{y}$, and each $z^l\succeq^I \bar{y}$.  By definition of bests', there is some $l$ for which $\lambda^l > 0$ and there must also be some $\mu^k>0$ since $\bar{x}$ does not best $\bar{y}$. So we may write 
$\bar{x} = \sum_{k=1}^{K_0} \mu^k w^k + \sum_{l=1}^{L_0} \lambda^l z^l$,
$\mu^k$ and $\la^l> 0$ for $k=1,\ldots,K_0$ and $l=1,\ldots,L_0$; and $\sum_{k=1}^{K_0} \mu^k + \sum_{l=1}^{L_0} \lambda^l =1$.

Consider first any $w^k$ for which the indirect preference is merely a consequence of $w^k\geq \bar x$. In other words, there is no observed data point $w'$ with $w^k\geq w'$ and $w'\succeq^I \bar x$. Without loss, suppose that this is $w^k=w^1$. Then note that $w^1> \bar x$ and thus $\mu^1<1$; so we may consider $\Delta=w^1-\bar x> 0$ and represent $\bar x$ as
\[
\bar{x} =\mu^1 [\bar x + \Delta]
+ \sum_{k=2}^{K_0} \mu^k w^k + \sum_{l=1}^{L_0} \lambda^l z^l
= \frac{1}{1-\mu^1}\left( \sum_{k=2}^{K_0} \mu^k w^k + \lambda^1 [z^1 +\frac{\mu^1\Delta}{\la^1}]
+ \sum_{l=2}^{L_0} \lambda^l z^l \right)
\]
Then $z^1\succeq^I \bar y$ and $\Delta\geq 0$ implies that  $z^1 +\frac{\mu^1\Delta}{\la^1}\succeq^I \bar y$, and we have reduced the number of $w^k\succeq^I \bar x$ by one. We may then assume that for each $w^k$ there exists some sequence $w',\ldots,w^*$ in the data so that $w^k\geq w'\succeq^R\ldots w^*\succeq^R \bar x$.

Consider now the set $\mathcal{W}$ consisting of the bundles that are 1) revealed indirectly preferred to $\bar{x}$, in the sense that there is an observed $w'$ with $w^k\geq w'\succeq^I \bar x$, and 2) not revealed indirectly preferred to $\bar y$. The bundles in $\mathcal{W}$ may not be in the support of $\bar x$, but $\mathcal{W}$ includes $w^1,\ldots,w^{K_0}$.

We claim that for any $w^k\in \mathcal{W}$ in the support of $\bar{x}$, there exists $w^{k'}\in \mathcal{W}$, also 
in the support of $\bar{x}$, for which $w^k \succ^I w^{k'}$ (in particular if $k$ is
unique then $w^k \succ^I w^k$). The claim provides a contradiction because it implies
the existence of a strict  $\succ^I$ cycle amongst the elements $w^k$, contradicting that the original data were rational.

To prove the claim, let $w^k\in \mathcal{W}$ in the support of $\bar{x}$ be arbitrary. Note that, if $w'$ is in the data, then $w'\geq z$ implies that $w'\succeq^R z$. So we may assume the existence of $w',\ldots,w^*$ with  $w^k \geq w'
\succeq^R \ldots \succeq^R w^* \succeq^R \bar{x}$, where all members of the chain are
members of $\mathcal{W}$ (as otherwise $w^k \succeq^I \bar{y}$, which we assumed
false by the definition of $w^k$). Note that the observed bundle $w^*$ is part of an observation $(p^*,w^*)$, so that $p^* \cdot w^* \geq p^* \cdot\bar{x}$.

Recall that there is at least one $z^l$, and that, for all $z^l$, $p^* \cdot w^* < p^* \cdot z_l$ (the latter as otherwise we would have $w^* \succeq^I z^l$, implying $w^* \succeq^I \bar{y}$ and hence $w^k \succeq^I \bar{y}$, again contradicting the definition of $w^k$).

So we have $p^* \cdot w^* \geq p^*\cdot (\sum_k \mu^k w^k + \sum_l \lambda^l z^l)$, and $p^* \cdot \lambda^l z^l > p^* \cdot \lambda^l w^*$ for all $l$, so that there must be $k'$ for which $p^* \cdot w^* > p^* \cdot w^{k'}$.  Conclude $w^* \succ^R w^{k'}$ and hence $w^k \succ^I w^{k'}$.  This then implies that there is a $\succ^I$ cycle of length at least two, contradicting the fact that GARP is satisfied.

Finally we show that strict besting' implies strict besting.  Suppose then that  $\bar{x}$ strictly bests' $\bar{y}$. We may write
$\bar{x} = \sum_k \mu^k w^k + \sum_l \lambda^l z^l$, with $z^l\succeq^I \bar y$ for
all $l$, and $w^k\succ^I \bar x$ for all $k$. By the previous proof, we also have
$w^k\succeq^I \bar y$. In fact, since $\bar x$ bests' $\bar y$ we can write $\bar x$ as a convex combination $\bar x = \sum_h \eta^h  r^h$ with each $r^h\succeq^I \bar y$.

Now consider $w^1$. First, if $w^1\gg \bar x$ then $\mu^1<1$ and we may proceed as above to eliminate $w^1$ from the representation of $\bar x$. Second, if $w^1\succ^I \bar x$ but it's not the case that $w^1\gg \bar x$ then by definition of $\succ^I$ there exists $w^*$ with
$w^k\succ^I w^*$ and $(p^*,w^*)$ is part of the data, with
\[
p^*\cdot w^*\geq p^*\cdot \bar x = p^*\cdot ( \sum_h \eta^h
r^h).\] The latter implies that
$p^*\cdot w^*\geq p^*\cdot r^h$ for some $r^h$, and hence that $w^1\succ^I r^h\succeq^I \bar y$. Thus $\bar x$ strictly bests $\bar y$.
\end{proof}

We may now proceed with the proof of Theorem~\ref{thm:varianwelfare}. The starting
point is the system of linear inequalities introduced by
\cite{varian1982nonparametric} for this problem. Indeed, these are essentially
Varian's Fact 4 (\citet{varian1982nonparametric}). In Varian's terminology, $\bar y$ is revealed worse than $\bar x$ if and only if there is no solution  $q>0$ to the system of linear inequalities comprised by the following collection of inequalities: \begin{enumerate}
    \item $q\cdot \bar x \leq q\cdot x^k$ for all $k$ with $x^k\succeq^I \bar x$
    \item $q\cdot \bar x \leq q\cdot x^k$ for all $k$ with $x^k\succeq^I \bar y$
    \item $q\cdot \bar x < q\cdot x^k$ for all $k$ with $x^k\succ^I \bar x$
    \item $q\cdot \bar x < q\cdot x^k$ for all $k$ with $x^k\succ^I \bar y$
    \item $q \cdot \bar x \leq q \cdot \bar y$
\end{enumerate}

Note that each of the first four listed inequalities really describes multiple linear inequalities. For example, there is one inequality $q\cdot \bar x \leq q\cdot x^k$ for each observation  $(p^k,x^k)$ that satisfies $x^k\succeq^I \bar x$. 

The first and third inequalities require that no revealed-preference cycle arises if
we add the hypothesized price $q$ to support $\bar x$, meaning that we add the
observation $(q,\bar x)$ to the data. The remaining inequalities require that with
this hypothesized price, $\bar x$ is not revealed strictly preferred, either directly
or indirectly, to $\bar y$.  If these inequalities are satisfied, then there is a
price $q$ that supports $\bar x$ for which $\bar x$ is not revealed strictly
preferred to $\bar y$.  No matter which price we choose to support $\bar y$, it will
then never be the case that $\bar x$ is revealed strictly preferred to $\bar y$.  It
is known that Afriat's Theorem then allows the flexibility to choose a
rationalization where $u(y) \geq u(x)$ (see Fact 16 in \cite{varian1982nonparametric}).

Let us set up a matrix to capture this system, with one row for each of the
inequalities that are collected in 1-5 above. These rows are of the form $x^k- \bar
x\in\Re^{\ngoods}$ or $\bar y - \bar x\in\Re^{\ngoods}$. We want $q>0$ so there is
also one row for each  $q_h\geq 0$ inequality, and one row for the inequality that
$\sum_h q_h>0$. Consider a dual solution with weights $\ta^k\geq 0$ for each of the
inequalities involving $\bar x$,  $\eta^k\geq 0$ for the inequalities that involve
$\bar y$, and $\eta^{\bar y}$ for the 5th inequality.

We use a prime to distinguish revealed preference from strict revealed preference.  Let $\xi^h\geq 0$ be the dual variable for the $q_h\geq 0$ inequalities and $\xi^M\geq 0$ for the last $\sum_h q_h>0$ inequality. The dual then says, for each $h$, \[\begin{split}
\sum_{\{k:x^k \succeq^I \bar x \}} \ta^k (x^k_h-\bar x_h) + 
\sum_{\{k:x^k \succ^I \bar x \}} \ta'^k (x^k_h-\bar x_h) + 
\sum_{\{k:x^k \succeq^I \bar y \}} \eta^k (x^k_h-\bar x_h) \\
+ \sum_{\{k:x^k \succ^I \bar y \}} \eta'^k (x^k_h-\bar x_h) + \underbrace{\eta^{\bar y}(\bar y_h - \bar x_h)}_{\bigstar} + \xi^M = 0
\end{split}
\]

In an abuse of notation, we shall not distinguish between variables with and without
prime. The term indicated by $\bigstar$, with dual variable $\eta^{\bar y}$,
corresponds to equation 5. For ease of exposition, label $x^{K+1} = \bar y$ and
$\eta^{K+1}=\eta^{\bar y}$, so that inequality 5 becomes an inequality of type 2, and
we write $\eta^{\bar y}(\bar y_h - \bar x_h)=\eta^{K+1}(x^{K+1}_h - \bar x_h)$. 

Suppose first that $\xi^M>0$. Then we get that $\sum_k (\ta^k + \eta^k) x^k \ll \bar x \sum_k (\ta^k + \eta^k)$, which means that $\sum_k (\ta^k + \eta^k)>0 $ and that we may normalize so that $\sum_k \ta^k + \eta^k=1 $. Set $z^{k^*}\gg x^{k^*}$ for some $\ta^{k^*}+\eta^{k^*}>0$, and $z^k=x^k$ for all other $k\neq k^*$, so that $\bar x = \sum_k (\ta^k + \eta^k) z^k$ with $z^k\succeq^I \bar x$ or $z^k\succeq^I \bar y$ for each $k$, and where the comparison becomes $\succ^I$ for $k=k^*$. Notice that we can choose $k^*$ so that $\eta^{k^*}>0$ because if all the $\eta$ variables were zero we would have a certificate for the inequalities in 1 and 3 being infeasible; we know, however, that these are feasible.\footnote{Indeed, if we consider only the inequalities and 1 and 3, and if the dataset is rationalizable, then we may choose $q>0$ to support a rationalizing utility at $\bar x$. The resulting dataset, adding the observation $(q,\bar x)$, must be rationalizable.} We conclude then that $\bar x$ strictly best' $\bar y$.

If instead $\xi^M=0$ then we must have $\ta^k+\eta^k>0$ for some $k$ with either $x^k\succ^I \bar x$ or $x^k\succ^I \bar y$. Again this allows us to assume that $\sum_k \ta^k + \eta^k=1 $ and we get that $\sum_k (\ta^k + \eta^k) x^k \leq \bar x$. Again we obtain that $\bar x$ strictly best' $\bar y$. By Lemma~\ref{lem:besting} the theorem follows.

\subsection{Proof of Theorem~\ref{thm:welfare}}
We begin with the following lemma, which is stated in \cite{chambers2016revealed}, Remark 3.6.

\begin{lemma}Let $i\in N$.  Suppose that for all $k\in[K_i]$, there are $u_i^k \in\Re$ and $\lambda_i^k > 0$ for which for all $k,l\in[K_i]$ satisfying $p_i^k \cdot (x_i^l - x_i^k)\leq 0$, we have \[u_i^l \leq u_i^k + \lambda_i^kp_i^k\cdot(x_i^l-x_i^k).\]  Then the individual dataset $\{(p_i^k,x_i^k)\}_{k\in [K_i]}$ is rationalizable. \end{lemma}

\begin{proof}Suppose that the condition in the statement of the Lemma is satisfied.  Define the pair of binary relations $x_i^k \succeq_i^R x_i^l$ if $p_i^k \cdot (x_i^l - x_i^k) \leq 0$ and $x_i^k \succ_i^R x_i^l$ if $p_i^k \cdot (x_i^l - x_i^k) < 0$.  

A \emph{cycle} is a finite list $x_i^{l_1} \succeq_i^R x_i^{l_2} \succeq_i^R \ldots x_i^{l_a} \succ_i^R x_i^{l_1}$.  We claim that there can be no cycle.  For, if there were, then we would have: \[u_i^{l_{j+1}}-u_i^{l_j}\leq \lambda_i^{l_j}p_i^{l_j}\cdot (x_i^{l_{j+1}}-x_i^{l_j}),\] for all  $j=1,\ldots,a-1$ and 
\[u_i^{l_{1}}-u_i^{l_a}\leq \lambda_i^{l_a}p_i^{l_a}\cdot (x_i^{l_{1}}-x_i^{l_a}).\]

Reading addition of indices as modulo $a$, observe that \[0=\sum_{j=1}^a (u_i^{l_{j+1}}-u_i^{l_j}) \leq \sum_{j=1}^a \lambda_i^{l_j}p_i^{l_j}\cdot (x_i^{l_{j+1}}-x_i^{l_j}) < 0.\]

The first equality is by telescoping, the weak inequality by summing the original inequalities, and the strict inequality because of the right hand sides of the original inequalities are nonpositive (and at least one strictly negative).  So, we arrive at a contradiction and there can be no cycle.  Conclude by Afriat's Theorem \citep{afriat1967construction,chambers2016revealed} that the individual dataset is rationalizable. \end{proof}

Now we proceed with the proof of the theorem.

First, that~\eqref{it:wf1} implies~\eqref{it:wf3} follows because if $u_i$ are rationalizing monotone and explicitly quasiconcave utilities, then $z_i\succeq^I_i \bar x_i$ implies $u_i(z_i)\geq u_i(\bar x_i)$, and $z_i\succ^I_i \bar x_i$ implies $u_i(z_i)> u_i(\bar x_i)$. So when $y_i$ is a convex combination of bundles $z^l_i\succeq^I_i \bar x_i$ we must have that $u_i(y_i)\geq u_i(\bar x_i)$ by quasiconcavity of utility. Moreover, if $z^l_i\succ^I_i \bar x_i$ for some $l$ then we obtain $u_i(y_i)> u_i(\bar x_i)$ by explicit quasiconcavity. In all, then, when $y_i$ bests $\bar x_i$ for all agents, and strictly bests for at least one agent, we have that $\bar x$ is Pareto dominated for the rationalizing utilities.

Second, it is obvious that~\eqref{it:wf2} implies~\eqref{it:wf1}. So we focus our attention on showing that~\eqref{it:wf3} implies~\eqref{it:wf2}. (Indeed our argument shows that~\eqref{it:wf2} and~\eqref{it:wf3} are equivalent.)  Suppose then that~\eqref{it:wf3} is satisfied. We will demonstrate that there exists some $q\in\Re_{++}^m$ so that, for all $i\in N$, the individual dataset given by $\{(p_i^k,x_i^k)\}_{k\in [K_i]}\cup \{(\overline{x}_i,q)\}$ is rationalizable.  This then implies (by Afriat's Theorem) the existence of a concave, increasing utility function for which for all $y\in\Re_+^m$ satisfying $q\cdot y \leq q \cdot \overline{x}_i$, we have $u_i(y)\leq u_i(\overline{x}_i)$, and consequently that $u_i(y)> u_i(\overline{x}_i)$ implies $q\cdot y > q\cdot \overline{x}_i$.  Consequently, it also follows that $u_i(y)\geq u_i(\overline{x}_i)$ implies $q\cdot y \geq q\cdot \overline{x}_i$, by continuity and monotonicity of $u_i$.  It then follows that $\overline{x}$ is efficient for these utility indices.\footnote{If not, then there is $\overline{y}$ for which $\sum_i \overline{y}_i = \sum_i \overline{x}_i$ and for all $i\in N$, we have $u_i(\overline{y}_i  )\geq u_i(\overline{x}_i)$, with inequality strict for some $j\in N$, implying $\sum_i q\cdot \overline{y}_i> \sum_i q\cdot \overline{x}_i$, a contradiction.}

The proof relies on a homogeneous Theorem of the Alternative: see \cite{BorderTofA}.  

The content of Afriat's Theorem is that  for each $i\in N$ and $k\in[K_i]$, there is $u_i^k$ and $\lambda_i^k>0$ for which for all $k,l\in[K_i]$, \[u_i^k\leq u_i^l + \lambda_i^l p_i^l \cdot (x_i^k-x_i^l).\]

What we would now like to find are additional unknown  parameters.  Namely, for each $i\in N$, a scalar $\overline{u}_i\in \Re$ and $q\in\Re^m$.   The vector $q$ is required to be common to all individuals and will reflect the common prices supporting the hypothesized efficient allocation $\overline{x}$. 

Our task is then to find $q\in\Re^m$, and for each $i\in N$, a real number $\overline{u}_i\in\Re$, and for each  $i\in N$ and $k\in [K_i]$, $u_i^k\in\Re$ and $\lambda_i^k\in\Re$ for which the following linear inequalities are satisfied:

\begin{enumerate}
\item \label{it:afriat1}
For all $i\in N$ and all $k,l\in[K_i]$ for which $p^k_i \cdot (x^l_i -x^k_i)\leq 0$, we have $u^l_i\leq u^k_i + \lambda^k_i p^k_i \cdot (x^l_i-x^k_i)$.
\item  \label{it:afriat2} For all $i\in N$ and all $k\in[K_i]$, $u_i^k \leq \overline{u}_i + q\cdot(x_i^k - \overline{x}_i).$
\item  \label{it:afriat3} For all $i\in N$ and all $k\in[K_i]$, for which $p_i^k\cdot (\overline{x}_i - x_i^k)\leq 0$, we have $\overline{u}_i \leq u_i^k + \lambda_i^k p_i^k\cdot (\overline{x}_i - x_i^k)$.
\item  \label{it:afriat4} For all $i\in N$ and all $k\in[K_i]$,  $\lambda_i^k > 0$.
\item  \label{it:afriat5} $q \geq 0$ and $q \neq 0$.
\end{enumerate}

The inequalities can be represented in matrix notation.  We display part of the matrix below, as the matrix itself is quite large.  The matrix below displays four horizontal blocks.  The first two correspond to vectors corresponding to weak inequalities, the latter two to strict. This matrix has, for each agent $i$, $2(K_i + 1)$ columns, and an additional $m$ columns; in total the number of columns is $m+\sum_i (2K_i + 1)$.  Observe that, in the matrix written below, the column labelled by $q$ actually represents $m$ columns; for example, $\mathbf{1}_{m'}$ is an indicator function of the dimension $m'\in\{1,\ldots,m\}$.

As to rows, the matrix has, for each agent $i$, one row for each ordered pair $(l,k)$ where  
$l,k\in[K_i]$, $k\neq l$, and $p_i^k\cdot (x_i^l-x_i^k)\leq 0$. When agent $i$ is understood, the row is labeled $(l,k)$, as in the displayed matrix below. Continuing with the rows for agent $i$, there are also three rows for each $k$: one labeled by $(k,*)$, one by $(*,k)$ and one by $k$. The row labeled $(k,l)$ for agent $i$ is meant to capture inequality~\eqref{it:afriat1}: there is a $1$ in the column $k$ for agent $i$, a $-1$ in column $l$, and $ p^k_i \cdot (x^l_i -x^k_i)$ in the column for $k$ among the second set of $K_i$ columns. The rest of the entries in that row are zero. In a similar vein, the rows labeled by $(k,*)$ and $(*,k)$ are there to encode the inequalities in~\eqref{it:afriat2} and in~\eqref{it:afriat3}. The row labeled $k$ is meant to capture the basic positivity constraint~\eqref{it:afriat4}, and has a one in column $k$, among the second collection of $K_i$ columns. 

Finally, the matrix has a collection of rows $m+1$ that are not specific to any agent and seek to capture~\eqref{it:afriat5}. There is then one column for each $m'\in\{1,\ldots,m\}$ (labelled $(*,m)$), expressing the nonnegativity of $q$, and a row asserting that $\sum_{m'=1}^m q(m')>0$; the row labelled $M$.  

Because this matrix is large, we only show certain portions of it.  The rows listed in the matrix have zeroes everywhere for every remaining column.

\[ \kbordermatrix{
     & 1 & \cdots& k  & \cdots& l & \cdots & K_i &  \cdots & * &\vrule     & 1' & \cdots & k'  &  \cdots & K_i' & \vrule & q\\
(l,k)& 0 & \cdots& 1 & \cdots& -1 & \cdots & 0   & \cdots & 0 & \vrule & 0 & \cdots & p^k_i\cdot(x^l_i-x^k_i) & \cdots & 0 &  \vrule& 0\\
\vdots & \vdots & & \vdots  & & \vdots &  & \vdots & & \vdots &\vrule     & \vdots &  & \vdots  &   & \vdots & \vrule & 0 \\ 
(*,k) & 0 & \cdots & 1 & \cdots & 0 & \cdots & 0 & \cdots & -1 & \vrule & 0 & \cdots & p^k_i \cdot (\overline{x}_i - x^k_i) & \cdots & 0 & \vrule & 0\\
\vdots & \vdots & & \vdots  & & \vdots &  & \vdots & & \vdots &\vrule     & \vdots &  & \vdots  &   & \vdots & \vrule & 0 \\ 
(k,*) & 0 & \cdots & -1 & \cdots & 0 & \cdots & 0 & \cdots & 1 & \vrule & 0 & \cdots & 0 & \cdots & 0 & \vrule & x_i^k - \overline{x}_i\\
\hline
\vdots & \vdots & & \vdots  & & \vdots &  & \vdots & & \vdots &\vrule     & \vdots &     & \vdots  &   & \vdots & \vrule & 0 \\
(*,m') & 0 & \cdots & 0 & \cdots & 0 & \cdots & 0 & \cdots & 0 & \vrule & 0 & \cdots & 0 & \cdots & 0 & \vrule & \mathbf{1}_{m'}\\
\vdots & \vdots & & \vdots  & & \vdots &  & \vdots & & \vdots &\vrule     & \vdots &     & \vdots  &   & \vdots & \vrule & 0 \\
\hline
M & 0 & \cdots & 0 & \cdots & 0 & \cdots & 0 & \cdots & 0 & \vrule & 0 & \cdots & 0 & \cdots & 0 & \vrule & \mathbf{1}_{\{1,\ldots,m\}}\\
\hline
\vdots & \vdots & & \vdots  & & \vdots &  & \vdots & & \vdots &\vrule     & \vdots &     & \vdots  &   & \vdots & \vrule & 0 \\
k & 0 & \cdots& 0  & \cdots& 0 & \cdots & 0 & \vdots & 0 &\vrule     & 0 & \cdots   & 1  &  \cdots & 0 & \vrule&0\\
\vdots & \vdots & & \vdots  & & \vdots &  & \vdots&&\vdots  &\vrule     & \vdots &     & \vdots  &   & \vdots & \vrule&0 \\
}\]

We are searching for a vector in $m+\sum_i (2K_i+1)$ dimensional real space which, when multiplied with this matrix to yield a linear combination of its columns, results in a vector whose coordinates in the first two horizontal blocks are nonnegative, and in the last two are strictly positive. Such a vector would represent a solution to the system of inequalities \eqref{it:afriat1}-\eqref{it:afriat5}. This is the system to which we will apply a duality result.

By Motzkin's transposition theorem (a version of the theorem of the alternative, see Theorem 47 in \cite{BorderTofA}) there is no solution to the set of inequalities (and consequently to the enumerated list of inequalities above) if and only if there is, for each row of the matrix, a nonnegative weight, where for some row corresponding to a strict inequality (either in the third or fourth horizontal block), one of the weights is strict, for which the weighted sum of rows is the zero vector.

So, let us suppose by means of contradiction that there is no solution to the linear system.  Therefore, there exists a solution to the dual system.  
Interpret the solution as a collection of weights on the rows of the matrix. 
For the rows corresponding to agent $i\in N$ (any row except the one labelled $M$), we let $\xi_i^A \geq 0$ denote the weight for the row labelled by $A$.  For example, in the row of the above matrix labelled $(l,k)$, $\xi_i^{(l,k)}$ is the associated weight.  We let $\xi^M \geq 0$ be the weight associated with row $M$ (which is common to all $i\in N$), and we let $\xi^{(*,m')}\geq 0$ be the weight associated with row $(*,m')$.



The matrix has a special structure.  Observe that, restricted to the first $\sum_i(K_i +1)$ block of columns on the left, and the rows labeled $(k,l)$, $(k,*)$, or $(k,*)$ for some agent (and some $k,l$), the matrix becomes the incidence matrix of a graph with vertexes that can be identified with these  $\sum_i(K_i +1)$ columns. So each vertex is identified with a pair $(i,k)$, of an agent and an observation $k\in [K_i]$, or with a pair $(i,*)$ for the hypothesized efficient bundle.  An edge goes from a node $(i,k)$ to $(i,l)$ when $p_i^k\cdot(x_i^l-x_i^k)\leq 0$.  An edge goes from $(i,*)$ to $(i,k)$ when $p_i^k\cdot(\overline{x}_i-x_i^k)\leq 0$.  An edge always goes from $(i,k)$ to $(i,*)$.  

Now, the solution to the dual, when restricted to the incidence submatrix, provides a non-negative linear combination of rows that equals the null vector. The Poincar\'e-Veblen-Alexander theorem \citep{berge2001theory} claims that for any non-negative weighted sum of incidence vectors of a directed graph which is zero, there is a collection of positively oriented cycles in the graph, each cycle being associated with a weight, and the total weight ascribed to an incidence vector is the sum of all weights associated to cycles in which the incidence vector appears.  Here, a cycle includes no repetitions of nodes.

Because the individual dataset $\{(p_i^k,x_i^k)\}_{k\in [K_i]}$ is rationalizable, we may assume without loss of generality that every such cycle involves an edge of the type connecting $(i,k)$ to $(i,*)$.  This is because otherwise, along all elements of the cycle, rationalizability implies that $p_i^{k_j}\cdot (x_i^{k_{j+1}}-x_i^{k_j})=0$, and thus the weighted sum of vectors across that cycle is zero.  Removing them does not affect the total weighted sum of rows.

Let us now represent the cycles associated with agent $i\in N$ by $\mathcal{C}_i$, as described, each of them comes  with a weight $\mu(c) \geq 0$.  What we just claimed is that for each $c\in \mathcal{C}_i$, there is some $k\in[K_i]$ and an edge connecting $(i,k)$ to $(i,*)$.  This implies, in particular, that $x_i^k \succeq_i^I \bar x_i$.  To see why, let the cycle be written via a sequence of nodes: $(i,*),(i,k_1),\ldots,(i,k_l=k),(i,*)$.  Because $(i,*)$ is connected to $(i,k_1)$ by an edge, it means that $p_i^{k_1}\cdot (\bar x_i - x_i^{k_1})\leq 0$, so that $x_i^{k_1} \succeq_i^R \bar x_i$; similarly, $x_i^{k_{j+1}}\succeq_i^R x_i^{k_{j}}$ for all $j=1,\ldots,l-1$.  Consequently, by definition, $x_i^k \succeq_i^* \bar x_i$.

What we have just claimed is that if $\xi_i^{(k,*)} > 0$, it must be that $x_i^k \succeq_i^I \bar x_i$. 





Now, again by Motzkin's transposition theorem, one of the following must be true:  either $\xi^M > 0$, or there is  $i\in N$ and $k\in[K_i]$ for which $\xi_i^k > 0$.

Let us consider each of the two cases in turn.  

\textbf{Case 1: There is a dual solution with $\xi^M>0$.}

The only columns for which row $M$ are nonzero are the last $m$ columns.  Rows of type $(*,m')$ add (potentially) non-negative terms to these last $m$ columns.  Since the weighted sum of rows equals zero, it follows that 
\begin{equation}\label{eq:equation}
\sum_i\sum_{k\in [K_i]} \xi^{(*,k)}_i(x^k_i - \overline{x}_i) = - \sum_{m'=1}^m \xi^{*,m'} \one_{m'} - \xi^M\one_{1\,\ldots,m}\ll 0.
\end{equation}
In other words,  for each $i\in N$ and each $k\in[K_i]$, there is a number $\theta_i^k\geq 0$ for which \[\sum_i\sum_{k\in [K_i]}\theta_i^k (x^k_i-\overline{x}_i)\ll 0,\] where by the preceding discussion, $\theta_i^k > 0$ implies $x_i^k \succeq_i^I \overline{x}_i$.  Furthermore, there is $i\in N$ and $k\in[K_i]$ for which $\theta_i^k > 0$, since equation~\eqref{eq:equation} is strictly negative in every coordinate.

Without loss of generality (since the system is homogeneous), we may assume that $\sup_{i\in N}\sum_{k\in [K_i]}\theta_i^k = 1$.  

For each $i\in N$, let $\ta^0_i = 1 - \sum_{k\in [K_i]} \ta^k_i$. Then \[ 
\sum_i (\ta^0_i \bar x_i + \sum_k \ta^k_i x^k_i ) 
= \sum_i (\bar x_i + \sum_k \ta^k_i (x^k_i - \bar x_i) ) 
\ll  \sum_i \bar x_i.
\]
So we can define
\[ 
\bar y_i = \ta^0_i \bar x_i  +  \sum_{k\in [K_i]}\ta^k_i x^k_i.
\] for all $i\neq 1$. Observe that $\bar y_i$ is a convex combination of $\bar x_i\succeq_i^I \bar x_i$ (by definition), and $x^k_i\succeq^I_i \bar x_i$. If $\ta^0_1>0$, choose  $y'_1\gg \bar x_1$ so that 
$\bar y_1 = \ta^0_1 y'_1  +  \sum_{k=1}^{K_1}\ta^k_1 x^k_1$ and $y'_1\succ^I_1 \bar x_1$; otherwise choose $y^{k^*}_1\gg x^{k^*}_1$ so that  $\bar y_1 = \ta^0_1 \bar x_1  +  \sum_{k=1}^{K_1}\ta^k_1 x^k_1 + \ta^{k^*}_1(y^{k^*}_1 - x^{k^*}_1)$ and $y^{k^*}_1\succ^I_1 x^{k^*}_1$. Either way the allocation $\bar y_i$ bests $\bar x_i$ for all agents, and strictly bests it for agent~1.

\textbf{Case 2: There is a dual solution with $\xi_i^k>0$.}

 This means that there is  $i\in N$ and $k\in[K_i]$ for which $\xi_i^k > 0$. Fix such an $i^*\in N$ and a $k^*\in[K_i]$.  Because $\xi_M = 0$ is possible, we may only conclude in this case that $\sum_i \sum_{k\in [K_i]}\xi_i^{(*,k)} (x_i^k - \overline{x}_i) \leq 0$.  
 
On the other hand, we may conclude, since $\xi_{i^*}^{k^*}>0$, that there is also $l\in \{1,\ldots,K_{i^*}\}$ with $\xi_{i^*}^{(l,{k^*})}>0$ and  $p_{i^*}^{k^*} \cdot (x_{i^*}^l -x_{i^*}^{k^*}) < 0$; or in other words, $x_{i^*}^{k^*}\succ_i^R x_{i^*}^l$.  In particular, the edge $(i^*,k^*)$ to $(i^*,l)$  belongs to some $c\in\mathcal{C}_i$, which has a corresponding $\xi_{i^*}^{(*,k)}>0$; we may conclude then that $x_{i^*}^k \succ_{i^*}^I \bar x_{i^*}$.
 
Now $\sum_i \sum_{k\in [K_i]}\xi_i^{(*,k)} (x_i^k - \overline{x}_i) \leq 0$ implies that we can again as in Case 1 set $\ta^k_i=\xi_i^{(*,k)}$, assume without loss that $\sum_k \ta^k_i\leq 1$, and define $\ta^0_i=1-\sum_k \ta^k_i$. Then we may set $z^0_i=\bar x_i$ when $\ta^0_i>0$ and $z^k_i=x^k_i$ when $\ta^k_i>0$ and then we have (ignoring terms where $\ta^k_i=0$) 
\[ \sum_i \sum_{k=0}^{K_i} \ta^k_i z^k_i \leq \sum_i \bar x_i
\] so that if we define an allocation by  $y_i=\sum_{k=0}^{K_i} \ta^k_i z^k_i$, and recall that  $x_{i^*}^k \succ_{i^*}^I \bar x_{i^*}$, we conclude that the allocation $(y_i)$ empirically dominates $(\bar x_i)$. 

\subsection{Proof of Theorem~\ref{thm:individualwelfare}}

For this proof we start by constructing the same matrix as in the proof of Theorem~\ref{thm:welfare} but with $N=1$, and where we now add a  row $\one_{*}-\one_{k}$ for each $k$ to capture the inequality $u^k \leq \bar u$. The idea is to consider the same collection of linear inequalities as before, but where we in addition require that the level of utility in the new observation exceeds that of any existing observation in the data. Consider a solution to the dual. Again when restricted to the incidence matrix there is a collection of oriented cycles in the graph, each cycle being associated with a weight, and the total weight ascribed to an incidence vector is the sum of all weights associated to cycles in which the incidence vector appears.  A cycle includes no repetitions of nodes.

Because the individual dataset $\{(p_i^k,x_i^k)\}_{k\in [K_i]}$ is rationalizable, we may assume without loss of generality that every such cycle involves an edge of the type connecting $(i,k)$ to $(i,*)$.  This is because otherwise, along all elements of the cycle, rationalizability implies that $p_i^{k_j}\cdot (x_i^{k_{j+1}}-x_i^{k_j})=0$, and thus the weighted sum of vectors across that cycle is zero.  Removing them does not affect the total weighted sum of rows.

By the same argument as in Theorem~\ref{thm:welfare}, if $\mathcal{C}$ denotes the set of cycles, each of them with weight $\mu(c)$, we know that a cycle has an edge connecting (say) $(k)$ to $(*)$, where $\xi^{(k,*)}>0$ and that in consequence $x^k\succeq^I \bar x$. What is different from the proof of Theorem~\ref{thm:welfare} is that now the cycle may involve an edge going from (say) $(l)$ to $(*)$ which was added from a row $\one_{*}-\one_{l}$ due to the inequality $u^l\leq \bar u$.

Now as before there are two cases to contend with. First, when $\xi^M>0$ we obtain as before that $\sum_k \xi^{(k,*)}(x^k-\bar x)\ll 0$. This means that there is a convex combination $\ta^-\bar x + \sum_k \ta^k x^k \ll \bar x$ with support in $\bar x$ and the $x^k\succeq^I \bar x$ (as $\ta^k=\xi^{(k,*)}>0$ means that the argument in previous paragraph applies). Second, when $\xi^M=0$ then we must have $\xi^k>0$ for some $k$. This may again lead to the same case as in Theorem~\ref{thm:welfare}, or it may be the case that $\xi^{(k,*)}=0$ for all $k$ and we have a strict cycle involving the new $\bar x\succeq^R x^l$ edges. This would be a violation of GARP.

\subsection{Proof of Theorem~\ref{thm:QL}}\label{sec:proofthmQL}
We offer only a sketch, as the details are similar to our other results.

For all $i\in N$ and all $k\in [K_i] \cup [L]$ and $t\in [K_i]$, consider the Afriat inequalities:
\[
U_i^k \leq U_i^t + p_i^t\cdot (x_i^k-x_i^t).
\]  In these inequalities, $U_i^k$ is unknown.

For all $k\in [K_i] \cup [L]$ and all $s\in L$, consider the Afriat inequalities for the  $L$ proposed
allocations,
\[
U_i^k \leq U_i^s + q^s\cdot (x_i^k-x_i^s).
\] 
In these inequalities, $U^k_i$ is unknown for $k\in [K_i]\cup [L]$ and $q^s\in\mathbb{R}^n_+$ is unknown for $s\in [L]$. 

Consider three matrices, $A,B,C$. These matrices have one row for each triple $(i,k,t)$ with $i\in N$, $k,t\in [K_i]\cup [L]$ and $k\neq t$.

Matrix $A$ has one column for each element of $(\cup_{i\in N} [K_i])\cup ([n]\times [L])$: identify each column with the unknown $U_i^s$. 
In the row for $(i,k,s)$ $A$ is equal to zero everywhere except for a $-1$ in the column for $U^k_i$ and $1$ in the column for $U^s_i$. 

Matrix $B$ has $n\times L$ columns: identify each with the unknown $q^s_\ell$. In the
row for $(i,k,s)$ in which $s\in [L]$ matrix $B$ has zero in all entries except for a $x^k_{i,\ell}-x^s_{i,\ell}$ in the column for $q^s_\ell$.

Matrix $C$ has a single column. In the row for $(i,k,s)$ with $s\in [K_i]$ this column equals $p^s_i\cdot (x^k_i-x^s_i)$. It equals zero in any row $(i,k,s)$ with $s\in [L]$.

  For each row $r=(i,k,s)$ and matrix $a\in\{A,B,C\}$ we write $r_a$ for row $r$ in
  matrix $a$. The system is infeasible iff there exists weights $\ta(r)\geq 0$ for each row $r=(i,k,s)$ such that
  \begin{enumerate}
  \item $\sum_{r}\ta(r)r_A = 0$
  \item $\sum_{r}\ta(r)r_B \leq 0$
  \item $\sum_{r}\ta(r)r_C <0$
  \end{enumerate}

  Note that, for each $i$ and $k$, $\sum_{r}\ta(r)r_A = 0$ implies that
  \[\sum_{s\in [K_i]\cup [L],s\neq k} \ta(i,k,s) - \sum_{s\in [K_i]\cup [L],s\neq k} \ta(i,s,k) =0.\]

Let $\eta_i(k,s)=\ta(i,k,s)$ and define $\eta_i(k,k)$ for each $k$ so that the matrix $\eta$ has constant row-column sum.
  
Since $\sum_{r}\ta(r)r_B \leq 0$ and $\eta_i$ is independent of $\ell$ we obtain
that, for each $s\in L$,  $\sum_{i\in N}\sum_{k\in [K_i]\cup [L]} \eta_i(k,s)
(x^k_{i}-x^s_{i})\leq 0$. Finally, $\sum_{r}\ta(r)r_C <0$ implies that 
\[\sum_i \sum_{t\in K_i}\sum_{k\in [K_i]\cup [L]}\eta_i(k,t)p_i^t\cdot(x_i^t-x_i^k)>0.\]

\subsection{Proof of Proposition~\ref{prop:CRU}}

The result follows from two simple lemmas.

\begin{lemma}\label{lem:CRU1}
\[  S^{u^*_i}(x^K_i)
  = \bigcap\limits_{(V_i,\la_i)\in A_i}S^{u^{(V_i,\la_i)}}(x^K_i)\]
\end{lemma}
\begin{proof}
Suppose that $u^{(V,\la)}(x_i)\geq  u^{(V,\la)}(x^K_i)$ for all $(V,\la)\in
A_i$. Then since $u^{(V,\la)}(x^K_i)=V^K_i=1$ for all  $(V,\la)\in
A_i$, it follows that $u^*_i(x_i) = \inf\{u^{(V,\la)}(x_i):(V,\la)\in A_i \}\geq u^*_i(x^K_i)$.

Conversely, suppose that $u^*_i(x)\geq u^*_i(x^K)$. Then again, since
$u^*_i(x^K)=1=u^{(V_i,\la_i)}(x^K)$ for any $(V_i,\la_i)\in A_i$,
we conclude that  $u^{(V_i,\la_i)}(x_i)\geq  u^{(V_i,\la_i)}(x^K_i)$ for all $(V_i,\la_i)\in
A_i$.
\end{proof}

\begin{lemma}\label{lem:CRU2}
Let $u_i$ be concave and monotone rationalization of the data with $u(x^K_i)=1$. Then
there is $(V_i,\la_i)\in A_i$ such that
$S^{u_i}(x^K_i)\subseteq S^{u^{(V_i,\la_i)}}(x^K_i)$.
\end{lemma}
\begin{proof}
For each $x_i\in\Re^{\ngoods}_+$, let $V_i^{x_i} = u_i(x_i)$ and $q_i^{x_i}\in\partial u_i(x_i)$. Then we
have, for any $x_i$ and $y_i$ that $V_i^{y_i}\leq V_i^{x_i}+q_i^{x_i}\cdot (y_i-x_i)$. We also have that
\[
u_i(x_i)=\inf\{V_i^{y_i}+q_i^{y_i}\cdot (x_i-y_i):y_i\in\Re^\ngoods_+ \}.
\]

Let $(p_i^k,x_i^k)$, $k=1,\ldots,K$ be a dataset.

If $u_i$ rationalizes the data, then we can identify $V_i^{k}=V_i^{x^k}$ and choose
$\la_i^k$ so that 
$q_i^{x^k}= \la_i^k p_i^k$. Because $u$ is a rationalization, then, $(V_i,\la_i)\in A_i$.  The resulting Afriat utility satisfies that, for any
$x_i\in S_{u_i}(x_i^K)$,
\begin{align*}
u^{(V_i,\la_i)}(x^K_i)=V^K_i & = u_i(x^K_i) \leq u_i(x_i)  \\
& = \inf\{V^y_i + q^y_i\cdot (x_i-y_i) : y_i\in\Re^\ngoods_+\}  \\
& = \inf\{V^k_i + q^k_i\cdot (x_i-x^k_i) :k=1,\ldots,K\}  \\
\end{align*}
 Hence $S_{u_i}(x^K_i)\subseteq S_{u^{(V_i,\la_i)}}(x^K_i)$.
\end{proof}

\clearpage
\bibliographystyle{econometrica}
\bibliography{efficiency}
\end{document}